%% file: WPT_Ant_Assign_V2.tex
\documentclass[conference]{IEEEtran} 
\usepackage{fancyhdr}
\usepackage{epsfig}
\usepackage{threeparttable}
\usepackage{epsf,epsfig}
\usepackage{amsthm}
\usepackage[cmex10]{amsmath}
\usepackage{amssymb, amsfonts}
\usepackage[noadjust]{cite}
\usepackage{dsfont}
\usepackage{subfigure}
\usepackage{color}
\usepackage{soul}
\include{header}

\usepackage{epstopdf}
\begin{document}

\title{Receiver Antenna Partitioning for Simultaneous Wireless Information and Power Transfer}
\author{
\authorblockN{Rahul Vaze}
\authorblockA{
School of Tech. \& Computer Science\\
Tata Institute of Fundamental Research\\
Mumbai, India\\
Email: vaze@tcs.tifr.res.in}
\and
\authorblockN{Jainam Doshi}
\authorblockA{
Dept. of Electrical  Engr. \\
Indian Institute of Technology, Madras\\
Chennai, India \\
Email: jainam.jainam@gmail.com}
\and
\authorblockN{Kaibin Huang}
\authorblockA{
Dept. of Electrical \& Electronic Engr. \\
The University of Hong Kong\\
Pok Fu Lam, Hong Kong \\
Email: huangkb@eee.hku.hk}}

\maketitle
\begin{abstract}
Powering mobiles using microwave \emph{power transfer} (PT)  avoids the inconvenience of battery recharging by cables and ensures uninterrupted mobile operation. The integration of PT and \emph{information transfer} (IT) allows wireless PT to be realized by building on the existing infrastructure for IT and  also leads to compact mobile designs. As a result, \emph{simultaneous wireless  information and power transfer} (SWIPT) has emerged to be an active research topic that is also the theme of this paper. In this paper, a practical SWIPT system is considered where two multi-antenna  stations perform separate PT and IT to a multi-antenna mobile to accommodate their difference in ranges. The mobile dynamically assigns each antenna for either PT or IT. The antenna partitioning results in a tradeoff between the MIMO IT channel capacity and the PT efficiency.  The  optimal partitioning  for maximizing the IT rate under a PT constraint is a NP-hard integer program, and the paper proposes solving it via efficient greedy algorithms with guaranteed performance. To this end, the antenna-partitioning problem is proved to be one that optimizes a sub-modular function over a matroid constraint. This structure allows the application of two well-known greedy algorithms that yield solutions no smaller than  the optimal one scaled by factors $(1-1/e)$ and $1/2$, respectively. 
\end{abstract}
\section{Introduction}

The arguably  most desirable new feature for mobile devices is wireless power transfer (PT), which eliminates the need of recharging using cables and avoids interruptions of mobile services due to dead batteries. With rapid advancements in microwave technologies, microwave PT has emerged to be a promising solution for wirelessly  powering mobiles due to its long transfer ranges (up to hundreds of meters) and support of mobility \cite{Huang:CuttingLastWiress:2014}. In contrast, non-radiative technologies for wireless PT e.g., inductive coupling and resonant coupling, suffer from extremely short ranges (less than a meter). Using microwaves as carriers, wireless PT and information transfer (IT) can be seamlessly integrated, which has resulted in the emergence of an active research area called simultaneous wireless information and power transfer (SWIPT) \cite{Varshney:TransportInformationEnergy:2008, Zhang:MIMOBCWirelessInfoPowerTransfer, Huang:CuttingLastWiress:2014}. The research on  SWIPT, however, requires 
thorough revamping of classical theories for wireless communications and networking  to achieve not only high IT rates but also high PT efficiencies. 

PT and IT concern two different aspects of data bearing microwaves, namely their information content and absolute power, respectively. As a result, PT can tolerate  much less propagation loss and support much shorter transmission distances than IT.  Furthermore, depending on the channel and energy states, a mobile may choose to operate in either the IT, PT, or SWIPT modes. Consideration of such factors in realizing SWIPT calls for the design of new algorithms/protocols  for MIMO transmissions \cite{Zhang:MIMOBCWirelessInfoPowerTransfer},  multiple access \cite{Zhang:ThputMaxWirelessPowerCommNet}, resource allocation \cite{HuangLarsson:SIPTBroadbandChannel, NgLo:MultiuserOFDMSInfoPowerTransfer},  mobile transceivers \cite{Zhang2013SWIPT}  and network architectures \cite{HuangLauArXiv:EnablingWPTinCellularNetworks:2013}.  

A simple design of a  SWIPT enabled mobile receiver  is to combine a conventional information receiver  and an RF energy harvester. The form factor of this design can be reduced by sharing antennas between the receiver  and harvester where the output of each antenna is split for data processing and energy harvesting \cite{ZhouZhang:WlessInfoPowrTransfer:RateEnergy:2012}. However, the addition of a  power splitter with an adjustable splitting ratio for each antenna increases the receiver  complexity. A simpler SWIPT-receiver  design that allows 
antenna sharing but requires no splitting, is to partition the set of antennas into two sets, one dedicated for IT and other for PT. This design builds on the classic antenna selection technique for MIMO communications \cite{HeaSanETAL:AnteSeleSpatMult:Apr:01} and requires a small number of RF chains, leading to a high-efficient mobile design. The problem of optimal antenna assignment/partitioning for the special case of SWIPT with a single-input-multiple-output IT channel has been explicitly solved in \cite{Zhang2013SWIPT} for a simplified objective function. However, the problem for the general case with a MIMO IT channel is much more challenging to solve that depends on the  eigenmodes of the channel matrix. To be specific, the problem is a NP-hard  integer program. 

The main contribution of this paper is connect the general SWIPT antenna-partitioning problem to the rich field of  efficient sub-optimal integer-programming algorithms with guaranteed performance. This important connection is established by analyzing the structure of the antenna-partitioning problem. Specifically, the problem is shown to be equivalent to maximizing a sub-modular function with a matroid constraint. For a sub-modular function, the incremental gain of adding a new element diminishes with increasing set size. The proven structure allows two well-known greedy algorithms to be applied for solving the antenna-partitioning problem.  Moreover, the resultant solutions are shown to be equal to the optimal ones with the scaling factors no smaller than  $(1-1/e)$ and $1/2$. Simulation results reveal that the performance of the said antenna-partitioning algorithms substantially outperform the derived worst-case bounds.

\section{System Model}

Consider the SWIPT system in Fig. \ref{Fig:System} where a mobile is receiving information/data from the base station and power from a power beacon.
Let $N_t$, $N_r$, and $N_p$, denote the number of antennas at the base station, at the mobile, and at the power beacon, respectively. The MIMO channel from the base station to the mobile is represented by the complex $N_r\times N_t$ matrix $\bH$. Power is beamed from the power beacon to the mobile and the beamforming vector is denoted as $\bff$. Let $\bH'$ denote the MIMO channel from the power beacon to the mobile. The effective MISO channel after beamforming is defined as $\bg = \bH'\bff$.

Let $s_n\in\{0, 1\rbrace$ indicate whether the $n$-th receiver antenna of the mobile is assigned for IT ($s_n =1$) or PT ($s_n = 0$).  For ease of notation, the indicator variables are grouped into a vector $\bs = [s_1, s_2, \cdots, s_{N_r}]^\dagger$. Let $\bS$ be a $N_r \times N_r$ diagonal matrix with diagonal entries being the elements of $\bs$. It is assumed that the \emph{transmit channel state information} (CSIT) is unavailable at the base station. Thus transmission  power is equally allocated over all transmit antennas. The assumption is relaxed in Section~\ref{sec:CSIT} where the effect of power control is analyzed. With equal power allocation, the  IT channel capacity  can be written as 
\begin{equation}\label{Eq:IT:Cap}
C(\bs) = \log \det\left(\mathbf{I} + \frac{P}{N_t} \bS\bH\bH^\dagger\bS^\dagger\right). 
\end{equation}
Note that the effect MIMO IT channel matrix $\bS\bH$ in \eqref{Eq:IT:Cap} consists of the rows from $\bH$ corresponding to receiver  antennas assigned for IT.

Given the antenna partition  specified by $\bs$, maximum-ratio combining is applied at the energy harvester to  maximize the harvested  power, denoted $P_r$ and given as 
\begin{equation}
P_r = \sum_{n=1}^{N_r}  (1 - s_n) |g_n|^2
\end{equation}
where $g_n$ is the $n$-th element of the mentioned effective MISO PT channel  $\bg$. The power $P_r$ is required to exceed the threshold $p_c > 0 $ representing  fixed circuit-power  consumption, called the \emph{circuit-power constraint}.

\begin{figure}[t]
\centering 
\includegraphics[width=3.5in]{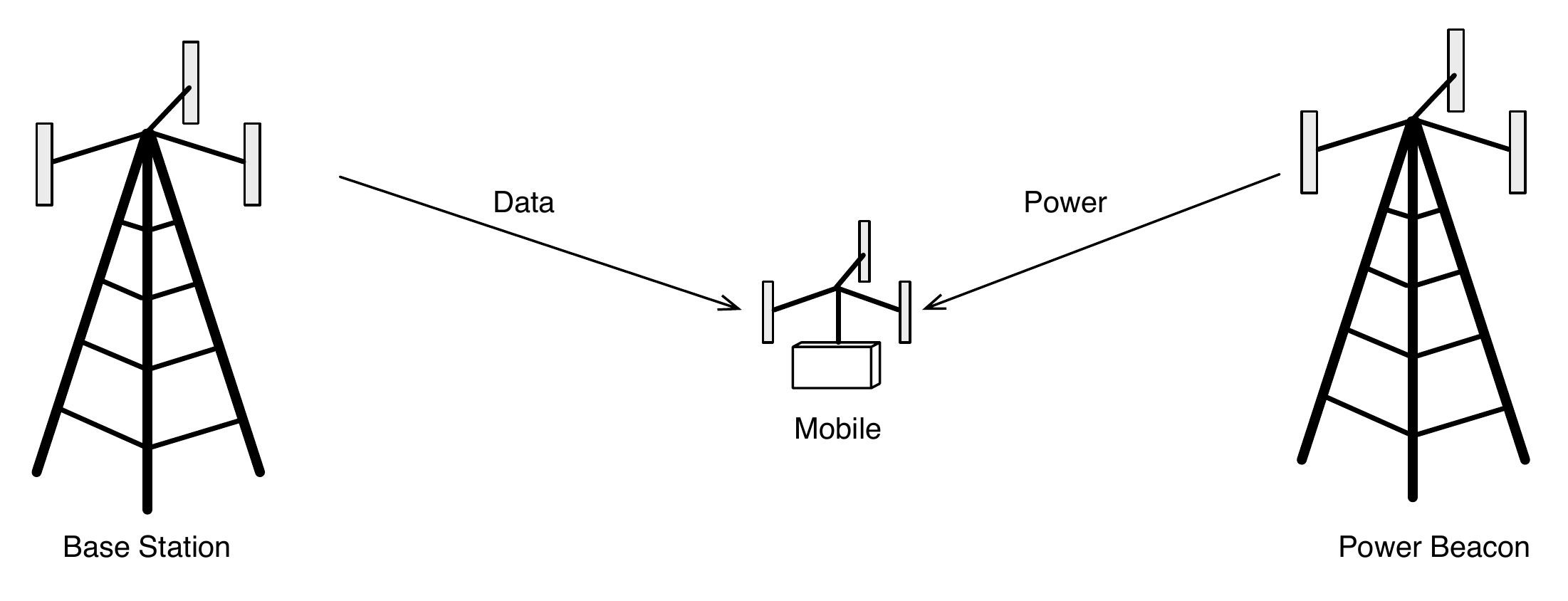}
\caption{A system supporting simultaneous wireless information and power transfer}
\label{Fig:System}
\end{figure}

\section{Problem Formulation}
The optimal antenna partitioning is formulated as the problem of maximizing the IT channel capacity in \eqref{Eq:IT:Cap} under the circuit-power constraint as follows: 
\begin{equation}
{\bf (P_1)} \qquad \begin{aligned}
    \underset{\{\bs\rbrace}{\text{max}} \quad & \log \det \left(\mathbf{I} + \frac{P}{N_t} \bS\bH\bH^\dagger\bS^\dagger\right) \\
\text{s.t.} \quad & s_n \in \{0, 1\rbrace, \quad n = 1, 2, \cdots, N_r\nn\\
\quad & \sum_{n=1}^{N_r}  (1 - s_n) |g_n|^2 \geq p_c. 
\end{aligned}
\end{equation}
This is an integer programming problem which is typically {\it NP}-hard to solve. However, the investigation of the structure of $P_1$ in the sequel leads to efficient methods for finding approximate solutions. 

\begin{remark} One sub-optimal method for solving $\bP_1$ is to relax the binary constraints in $\bP_1$ to allow $s_n \in [0, 1]\ \forall \ n$, 
that corresponds to allowing dynamic power splitting of the
received signal at all antennas for the purposes of information
transfer and power transfer as done in \cite{Zhang2013SWIPT} for SISO channel. Since the objective function is a concave function of $\bs$,   the approximate problem can be efficiently solved using  standard convex optimization algorithms. However, this method does not provide any guarantee on the performance and its degradation compared with the optimal one can be unacceptable for certain channel realizations. 
\end{remark}

\section{$(1-1/e)$-Approximate Solution}
A sub-optimal algorithm for solving $\bP_1$  can be designed to have guaranteed worst-case performance specified by an approximation ratio $a\in (0, 1]$  defined as follows. 

\begin{definition} An algorithm having  an approximation ratio of $\mathsf{a} \in  (0, 1]$ ensures that the ratio of the capacity \eqref{Eq:IT:Cap} evaluated at  its output $\hat{\bs}$ and the optimal solution $\bs^\star$ satisfies 
$\frac{C(\hat{\bs})}{C(\bs^\star)}\geqslant \mathsf{a}$, regardless of $\bH$. 
\end{definition}

\begin{remark} The approximation ratio $\mathsf{a}$ specifies the worst-case performance but the actual capacity  $C(\hat{\bs})$ can be substantially larger than the lower bound 
$\mathsf{a} C(\bs^\star)$. 
\end{remark}

To proceed further, we need some preliminaries. Let $f$ be a set function defined over all subsets of $U$: $f : 2^U \rightarrow \mathbb{R}^+$ where $2^U$ denotes the power set of $U$. 
\begin{definition}(Monotonicity)  The function  $f$ is  monotone if $$f(S \cup \{a\rbrace) - f(S) \ge 0,$$ for all $a\in U, S\subseteq U, a\notin S$. 
\end{definition}
\begin{definition}(Sub-modularity)  The function $f$ is {\it sub-modular} if  $$f(S \cup \{a\rbrace) - f(S) \ge f(T \cup \{a\rbrace) - f(T),$$  for all elements $a\in U, a\notin T$ and all pairs of subsets $S \subseteq T \subseteq U$. 
%In particular, a function $f$ is called {\it modular} if it satisfies $f(S \cup \{a\rbrace) - f(S) = f(T \cup \{a\rbrace) - f(T)$, for all elements $a\in U, a\notin T$ and all pairs of subsets $S \subseteq T \subseteq U$.
\end{definition}
Essentially, for a sub-modular function,  the incremental gain from adding an extra element in the argument  set decreases with the size of the set. 
The definition of sub-modular function is introduced for the  reason that the capacity function $C(\cdot)$ in \eqref{Eq:IT:Cap}  is sub-modular as shown in Lemma~\ref{lem:submodMI} that directly follows from  \cite[Theorem~$4$]{Vaze}.
\begin{lemma}\label{lem:submodMI} The function $C(\cdot)$ of the set of receive antennas assigned for IT is sub-modular.
\end{lemma}
It is easy to verify that $C(\cdot)$ is a monotone function since adding more receiver antennas cannot decrease the capacity.

Next, \emph{multi-linear extension} and \emph{matroid} are defined as follows. Let $n$ denote the cardinality of $S$.  Consider a function $f$ that assigns a nonnegative   value to each subset of $S$: $f: 2^S \rightarrow \mathbb{R}^{+}$. For each subset $A \subseteq S$, let $\bx_A = [x_1\dots x_n]$ be the $n$-length vector, where $x_i =1$ if the $i$th element of  $S$ is contained in $A$ and $x_i=0$ otherwise. Thus, $f$ is  a function assigning a value to each vertex of a $\{0,1\}^n$ hypercube. 

\begin{definition} (Multi-linear extension) The multi-linear extension $F$ extends $f$ to whole of $[0,1]^n$ such that for $\bx \in [0,1]^n$
$$F(\bx) = \E\{f({\hat \bx})\rbrace =\sum_{A \subseteq S} f(A) \prod_{i\in A} x_i \prod_{j \notin A} (1-x_j), $$
where  ${\hat \bx}$ denotes the random vector comprising $n$ i.i.d. Bernoulli random variables where the $j$th  variable is  $1$ with probability $x_j$ and $0$ otherwise.
\end{definition}
Note that $S$ in  the current  case is the set of receiver antennas $[1, \cdots, N_r]$.

\begin{definition} (Matroid) Consider a set $S$ of $n$ elements. A
 family $ {\cal I} \subseteq 2^S$
 of subsets of $S$ is called a matroid, denoted as  $\cM(S,{\cal I})$, and $\mathcal{I}$ is called a family of independent sets under the following two conditions: 
 \begin{enumerate}
\item if $X\subseteq Y$ and $Y\in {\cal I}$, $X \in {\cal I}$; 
\item if $X \in {\cal I}$ and $Y\in {\cal I}$ with $|X| \le |Y|$, there exits $ \{e\rbrace \in Y\backslash X$ 
such that $X\cup \{e\rbrace \in {\cal I}$.
\end{enumerate}
\end{definition}
\begin{lemma}\label{lem:matroid}The circuit power constraint of ${\bf P_1}$ is a matroid.
\end{lemma}
\begin{proof}
The circuit-power constraint of ${\bf P_1}$ is linear and equivalent to 
$$- \sum_{n=1}^{N_r}  (1 - s_n) |g_n|^2 \leq -p_c.$$ Thus, if this inequality holds 
for  a set of receive antennas $S$, clearly it also holds for  a subset $A \subseteq S$, thus satisfying condition $1)$ of matroid definition. Condition  $2)$ can also be verified immediately, completing the proof. 
\end{proof}

Consider a set $S$ of $n$ elements and a matroid $\mathcal{M} = (S,\mathcal{I})$ of subsets of $S$. Let $w_j$ for $j=1,2, \dots ,n$ be given weights of the elements of $S$.
 Let $T \in  \mathcal{I}$ and let $\bx_T = [x_1\dots x_n]$ be the $n$-length vector, where $x_i =1$ if the $i^{th}$ element of $S$ is contained in $T$ and $x_i=0$ otherwise. 
\begin{definition} (Weight of an independent set)  The weight of the subset $T$ is given by
 \begin{equation}
  w(T) = \sum_{i=1}^{n}w_ix_i. 
 \end{equation}
\end{definition}

With the above definitions and results, $\bP_1$ can be readily transformed into the following general problem of optimizing a sub-modular objective function over the convex hull of a matroid, which taps into the rich literature on algorithms for solving such problems: 
\begin{equation}
{\bf (P_2)} \qquad \max_{\bx \in \mathsf{conv}(\mathcal{M})} f(\bx)
\end{equation}
where $f$ is a sub-modular function and $\mathsf{conv}(\mathcal{M})$ represents the convex hull of $\mathcal{M}$. Specifically, $\mathsf{conv}(\mathcal{M}) = \mathsf{conv}\l(\{ \mathbf{1}_{I}: I \in {\cal I}\}\r)$ where $\mathbf{1}_{I}$ is the indicator vector of length $|S|$, where $[\mathbf{1}_{I}]_j=1$ if $j\in I$.  It follows from Lemma \ref{lem:submodMI} and \ref{lem:matroid} that $\bP_1$ can be written in the form of $\bP_2$.

In the remainder of this section, an approach is presented for computing an approximate solution for  $\bP_2$ (or equivalently $\bP_1$). This approach comprises three algorithms described in the sequel as follows. Consider a set $S$, a matroid $\mathcal{M} = (S,\mathcal{I})$, and the weights of the elements $\{w_j\}$. The first algorithm, Algorithm~$1$, was proposed in \cite{schrijiver} for maximizing  the weight of an independent set in $\mathcal{I}$, namely solving  the maximum the problem $\max_{I \in \mathcal{I}}w(I)$. The procedure for Algorithm~$1$ is presented as follows. 

\vspace{5pt}
\noindent {\bf Algorithm~$1$:} Weight Maximization Algorithm for an Independent Set. 

\begin{enumerate}
 \item Rearrange the elements of $S$ and obtain $S = \left\lbrace e_1,e_2 \dots ,e_n \right\rbrace$ such that $w_{e_1} \geq w_{e_2} \geq \dots \geq w_{e_n}$.
 \item Initialize $X \leftarrow \emptyset$.
 \item For $i = 1 ~ \text{to} ~ n$, do \\
 if $ \left( X + e_i \in \mathcal{M} \right)$ and $w \left( X + e_i \right) \geq w(X)$ then $X \leftarrow X + e_i$.
 \item Output $X$.
\end{enumerate}
\vspace{5pt}

\begin{lemma}[\cite{schrijiver}] Algorithm~$1$ solves the  maximum weight independent set problem $\max_{I \in \mathcal{I}}w(I)$. 
\end{lemma}

Using Algorithm~$1$ as a sub-algorithm, the following \emph{continuous greedy} algorithm as proposed 
 \cite{Vondrak2008optimal} yields an approximate solution for ${\bf P_2}$. 

\vspace{5pt}
\noindent {\bf Algorithm~$2$:} Continuous Greedy Algorithm
\begin{enumerate}
\vspace{.1in}
\item Let $\delta = \frac{1}{n^2}$, where $|S|=n$. Start with $t=0$ and $\bx(0) = (x_t(1) \dots x_t(n)) = {\bf 0}$.
\item Let $R_t$ be a vector of size $n$, where $R_t(j)=1$ independently with probability $x_t(j)$. For each $j \in S$, estimate 
$\omega_j(t) = \E\{f_{R(t)}(j)\rbrace$ say by taking the average of $n^5$ samples, where 
$$f_{R(t)}(j) = f(R(t)\ \cup \{j\rbrace) - f(R(t)).$$ 
\item  Let $I(t)$ be a maximum weight independent set in ${\cal M}$ computed by Algorithm~$1$, according to weights $\omega_j(t)$. Let $\bx(t+\delta) = \bx(t) + \delta. \mathbf{1}_{I(t)}$. 
\item Increment $t \bydef t+\delta$ if $t <1$, go to Step 2. Otherwise, return $\bx(1)$.
\end{enumerate}
\vspace{5pt}

As shown in  \cite{Vondrak2008optimal},  Algorithm~$2$ has the useful property that it has the guaranteed  worst-case performance as specified in the following lemma. 

\begin{lemma}\cite{Vondrak2008optimal}\label{lem:guaranteeCGA} The fractional solution $\bx$ of the optimization problem ${\bf P_2}$ found by Algorithm~$2$ satisfies 
$$F(\bx) = \E\{f({\hat \bx})\rbrace \ge \left(1-1/e\right)f(\bx^\star)$$
with high probability\footnote{A high probability  in this paper means one whose difference with $1$ diminishes exponentially in  $n$.}, where $\bx^\star$ gives the optimal value.\end{lemma}

The solution $\bx$ output by Alogrithm~$2$  may be \emph{fractional} defined as follows. Given a $y \in [0,1]^n$, we say that $i$ is fractional in $y$ if $0<y_i<1$, and for $y \in \cC(\mathcal{M})$, define $y(A) = \sum_{i\in A} y_i$. Then, a set $A \subseteq S$ is defined to be {\it tight} if $y(A) = r_{\mathcal{M}}(A)$, where $r_{\mathcal{M}}(A) = \max\{|I|: I \subseteq A \ \text{and}\  I \in \cI\}$ is the rank function of the matroid. For the case of fractional solution $\bx$ generated by Alogrithm~$2$, rounding each element of $\bx$ to be integers is necessary. The existence of an efficient algorithm for this purpose is shown in the following lemma from 
the result in \cite{calinescu2007maximizing}. 
\begin {lemma}
 Given a matroid ${\cal M} = (S, {\cal I})$, and a monotone sub modular function $f : 2^{S} \rightarrow \mathbb{R}^{+}$, and a fractional solution $x \in \cC(\mathcal{M})$, there exists 
 a polynomial time randomized algorithm, which returns an independent set $\cX \in \mathcal{I}$ of value $f(\cX) \geq (1-o(1))F(x)$, where $F$ is the multi-linear extension, and the $o(1)$ term can be made polynomially small in $n = |S|$. 
\end {lemma}

A particular  rounding algorithm, which was developed in  \cite{Ageev2004pipage, calinescu2007maximizing} and called the \emph{pipage rounding},  is used in this paper for rounding  $\bx$ to yield  an integer solution.   The detailed procedure is presented as Algorithm~$3$

\vspace{5pt}
\noindent {\bf Algorithm~$3$:} Pipage Rounding  Algorithm 

\begin{enumerate}
\vspace{0.1in}
\item Let $y$ denote the  fractional solution output by Algorithm~$2$; 
\item Find $A$, the minimal tight set containing at least $2$ fractional variables $i, j$;
\item Let $y_{ij}(\epsilon)$ be the vector obtained by adding $\epsilon$ to $y_i$, subtracting $\epsilon$ from $y_j$ and leaving the other values unchanged. Define $\epsilon_{ij}^{+}(y) = \max \l\{ \epsilon \geq 0 \mid y_{ij}(\epsilon) \in \cC(\mathcal{M})\r\rbrace$;
and $\epsilon_{ij}^{-}(y) = \min \left\{ \epsilon \leq 0 \mid y_{ij}(\epsilon) \in \cC(\mathcal{M})\r\rbrace$; 
\item If $F(y_{ij}(\epsilon_{ij}^+)) > F(y_{ij}(\epsilon_{ij}^-))$, then $y \leftarrow y_{ij}(\epsilon _{ij}^+)$ else $y \leftarrow y_{ij}(\epsilon _{ij}^-)$; 
\item If $y$ is fractional, go to step 1. Otherwise, return $y$.
\end{enumerate}
\vspace{5pt}

It is shown in the following lemma that rounding using Algorithm~$3$ does not compromise the worse-case performance of Algorithm~$2$. 

\begin{lemma}\cite{Ageev2004pipage} \label{lem:pipage} The integer solution $\bx_{int}$ for problem ${\bf P_2}$ obtained by applying pipage rounding to the fractional output $\bx$ of the continuous greedy algorithm satisfies  $f(\bx_{int}) \ge \left(1-1/e\right)f(\bx^\star)$ with high probability, where $\bx^{\star}$ is the optimal integer solution.
\end{lemma}

Finally, combining the above algorithms and their properties, the main result of this section is presented as follows. 
\begin{theorem}\label{thm:pipage} The pipage rounded solution $\bs_{int}$ of the fractional solution $\bs$ found by the continuous greedy algorithm for the antenna partitioning problem ${\bf P_1}$ satisfies $$C(\bs_{int})  \ge \left(1-1/e\right)C(\bs^\star)$$
with high probability, where $\bs^\star$ solves  ${\bf P_1}$.
\end{theorem}
\begin{proof} From Lemma \ref{lem:submodMI} and \ref{lem:matroid}, problem $\bP_1$ is a special case of problem $\bP_2$, and the result follows from Theorem \ref{thm:pipage}.
\end{proof}
\begin{remark}
Exploiting the sub-modularity of the IT channel capacity function   and matroidal circuit power constraint, Algorithms~$1$-$3$ give  a guaranteed  worst case approximation ratio of $\left(1-1/e\right)$ that is sufficiently close to $1$.  
However,  the complexity for these algorithms is high. The continuous greedy algorithm starts with $t=0$ and increments in the direction of the maximum weight independent set with a size of $\frac{1}{n^2}$ ($n=N_r$). In each increment, the weights of elements are found
by taking average of $n^5$ independent samples. Thereafter, finding the maximum weight independent set has a complexity of $O(n \log n)$. This results in an overall time complexity of $O(n^2(n^5(n) + n \log n)) = O(n^8)$. Subsequently, the pipage rounding algorithm takes $n^2$ iterations to convert the fractional solution given by the continuous greedy algorithm into an integral solution. Thus, the overall time complexity of the $1-1/e$ algorithm is $O(n^8)$. A more efficient approach is described in the next section that trades performance for low complexity. 
\end{remark}

\section{$1/2$-Approximate Solution}
A simple greedy algorithm for approximately solving ${\bf P_2 }$ is as follows. 

\vspace{5pt}
\noindent {\bf Algorithm~$4$:} Greedy Algorithm 

\begin{enumerate}

\item Start with set $s_0 = \mathbf{0}$;
\item At step $i$, $s_i = s_{i-1}+ [0 \dots 0 \underbrace{1}_{i^{\star}} 0\dots 0] $, where  \[i^{\star} =  \underset{i \in \{1,2,\dots, N_r \rbrace}{\operatorname{argmax}} ~ \log \det \left(\bI + \frac{P}{N_t}\bS_i\bH\bH^{\dag}\bS_i^\dag\right),\] where   $s_i = s_{i-1}+ [0 \dots 0 \underbrace{1}_{} 0\dots 0] $ and  $\bS_i$ is the diagonal matrix corresponding to $s_i$ as mentioned before; 
\item If $s_{i^{\star}}$ satisfies $\sum\limits_{j=1}^{n}(1 - s_{i^{\star}})|g_n|^2 \geq p_c$, repeat for $i = i+1$, else, output $s_{i-1}$. 

\end{enumerate}

The worse-case performance of the greedy algorithm is quantified as follows.  Let $S^{\star}$ be a set that maximizes the value of $f$ over all subsets of the matroid.  The following result is known for using greedy algorithms for maximizing monotone sub-modular functions  under a matroid constraint. 

\begin{theorem} \cite{Nemhauser1978}\label{thm:approx}   For a non-negative, monotone sub-modular
function $f$ and a matroid constraint, let subset $S$ be obtained by selecting elements one at a time, each time choosing an element that provides
the largest marginal increase in the function value that is feasible with respect to the  matroid constraint.
$f(S)\ge \frac{1}{2}f(S^{\star})$.\end{theorem}

\begin{corollary} The objective function of problem ${\bf P_1}$ evaluated at the greedy algorithm GA's output $\ge\mathsf{OPT}/2$, where $\mathsf{OPT}$ is the value of the optimal solution.
\end{corollary}
\begin{proof}\label{thm:greedyCSIR} From Lemma \ref{lem:submodMI}, we know that the objective function of ${\bf P_1}$ is monotone and sub-modular. Moreover, the linear circuit power constraint is a matroid from Lemma \ref{lem:matroid}. Thus, using Theorem \ref{thm:approx}, we have that the greedy algorithm gives a $1/2$ approximation to ${\bf P_1}$.
\end{proof}
\begin{remark}
If $S$ has $n$ elements, the $1/2$ approximate algorithm will terminate in a maximum of $n$ iterations, as a new element is included in every iteration. In each iteration, the element
with the highest marginal gain can be obtained by finding the maximum of $n$ entries. This results in a running time complexity of $O(n^2)$. Thus, even though, the continuous greedy algorithm gives a better bound on performance than the greedy algorithm, it has a running time complexity of $O(n^8)$ as compared to $O(n^2)$ for the greedy algorithm. Hence, both the algorithms have their own significance and which algorithm to be used depends on the problem at hand.
\end{remark}

\section{Extension to the Case of CSIT}\label{sec:CSIT}
In this section, the results in the preceding section are extended to the case of CSIT on the IT channel $\bH$ that enables the base station to allocate power over transmit antennas for increasing the IT channel capacity. 
Again,  let $s_n\in\lbrace0, 1\rbrace$ indicate whether the $n$-th receiver antennas at the mobile is assigned for information transfer ($s_n =1$) or power transfer ($s_n = 0$).  With CSIT, the power allocation (input covariance matrix) at different antennas of the base station depends on receiver antenna assignment vector $\bs$, and the capacity is modified from \eqref{Eq:IT:Cap} as 
\begin{eqnarray}\nn
C(\bs) &=& \max_{\bQ, \tr(\bQ)\le P}\log \det\left(\mathbf{I} + \bS\bH\bQ\bH^\dagger\bS^\dagger\right), \\\label{eq:objcsit}
& = & \max_{p_i, \sum_{i=1}^{N_r} p_i \le P} \sum_{i=1}^{N_r} \log(1+ s_i \lambda_i(\bs) p_i), 
\end{eqnarray}
where $\lambda_i(\bs)$ are the eigen-values of the submatrix $\bH(\bs)$ of $\bH$ which is obtained by keeping all rows of $\bH$ for which $s_i=1$, and without loss of generality we have assumed that $N_{\text{t}} \ge N_r$. The power allocation $p_i$ at the basestation depends on the receiver antennas allotted for data transfer, i.e., $p_i$ depend on $s_i$, and the optimal power allocation is given by waterfilling. Compared to \eqref{Eq:IT:Cap}, the expression inside the max in \eqref{eq:objcsit} is simple (sum of $N_r$ parallel channels), however, together with the max, the overall capacity expression is more complicated. Moreover, note that the choice of $s_i$ and $p_i$ depend on each other. Thus, the antenna partitioning problem in the CSIT case is more complex than the CSIR case. It follows that the corresponding problem of optimal antenna partitioning, denoted as  $\bP_3$,  is given as
\begin{equation}
{\bf (P3)} \qquad \begin{aligned}
    \underset{\lbrace\bs\rbrace}{\text{max}} \quad &  \max_{p_i, \sum\nolimits_{i=1}^{N_r} p_i \le P} \sum\nolimits_{i=1}^{N_r} \log(1+ s_i \lambda_i(\bs) p_i) \\ 
    \text{s.t.} \quad & s_n \in \lbrace0, 1\rbrace, \quad n = 1, 2, \cdots, N_r\nn\\
\quad & \sum_{n=1}^{N_r}  (1 - s_n) |g_n|^2 \geq p_c. 
\end{aligned}
\end{equation}

Even though the capacity expression \eqref{eq:objcsit} involves a maximization, it is still a sub-modular function as shown in \cite{VazeWaterfillingSubMod}. We summarize the result of \cite{VazeWaterfillingSubMod} as follows. 
\begin{theorem} For a set $S$, and fixed $s_i, i\in S$, the rate 
$$\bR(S) = \max_{p_i, \sum_{i\in S} p_i \le P} \sum_{i\in S} \log(1+ s_i \lambda_i(\bs) p_i),$$ obtained with a set $S$ of parallel Gaussian channels using the optimal waterfilling algorithm is a sub-modular function over the set of channels $S$. 
\end{theorem}

It is easy to show that \eqref{eq:objcsit} is a monotone function, since adding more receiver antennas cannot decrease the rate. Thus, we have that similar to the CSIR case, $\bP_3$ is a special case of problem $\bP_2$, and we get results on solving $\bP_3$ as described below.

\begin{theorem} The pipage rounded solution $\textbf{x}_{int}$ of the fractional solution $\textbf{x}$ found by the continuous greedy algorithm on $\bf{P_3}$ satisfies 
$$f(\bx_{int})  \ge \left(1-1/e\right)\mathsf{OPT}$$
with high probability, where $\mathsf{OPT}$ is the optimal value of the mutual information in  ${\bf P_3}$.
\end{theorem}

\begin{theorem}\label{thm:greedyCSIT} The objective function of problem $\bf{P_3}$ evaluated at the greedy algorithm GA's output $\ge\mathsf{OPT}/2$ for solving $\bP_3$, where $\mathsf{OPT}$ is the value of the optimal solution of $\bP_3$.
\end{theorem}
\section{Simulation Results} In this section, we illustrate the numerical performance of the two approximation algorithms presented in this paper to maximize the mutual information while satisfying the circuit power constraint. For the CSIR case, in Fig. \ref{fig:simCSIR}, we plot the throughput (mutual information) as a function of receiver antennas $N_r$ for fixed $N_t = 5$ and circuit power constraint of $0.2N_r$ with total transmit power $P= 5W = 6.9dB$ under a Rayleigh fading assumption on the channel matrix $\bH$. We scale the circuit power constraint with $N_r$ since larger the number of receiver antennas more is the power required for their operation. As can be seen from Fig. \ref{fig:simCSIR}, both the continuous greedy and the greedy algorithm perform better than their worst case bound of $1-1/e$ and $1/2$, respectively.  

In Fig. \ref{fig:simCSIT}, we plot the throughput as a function of receiver antennas $N_r$ for the CSIT case, with identical parameters used for Fig. \ref{fig:simCSIR}. Once again we see that the continuous greedy and the greedy algorithm perform better than their worst case bound of $1-1/e$ and $1/2$.

An important observation one can make from Fig. \ref{fig:simCSIR} and Fig. \ref{fig:simCSIT}, is that the greedy algorithm outperforms the continuous greedy + pipage rounding algorithm, even though, the worst-case performance guarantees of the  continuous greedy + pipage rounding algorithm is better than the greedy algorithm.

%The continuous greedy algorithm performs better than the greedy algorithm, however, its running complexity is also much higher and typically of the order $N_r^8$, while greedy algorithm runs in $N_r^2$ time.

\begin{figure}[htbp]
\centering 
\includegraphics[width=3.8in]{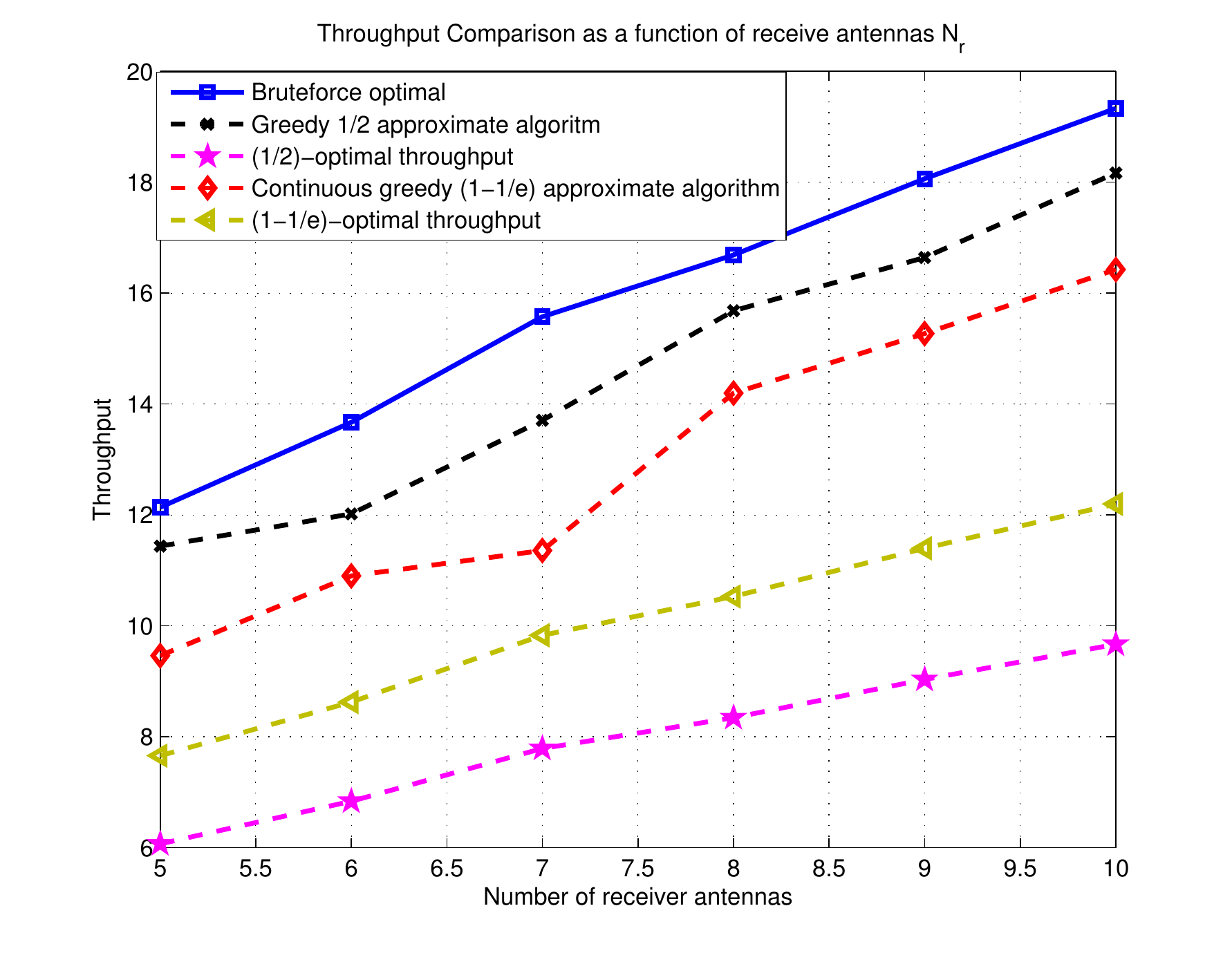}
\caption{Mutual information comparison with different algorithms for CSIR}
\label{fig:simCSIR}
\end{figure}

\begin{figure}[htbp]
\centering 
\includegraphics[width=3.8in]{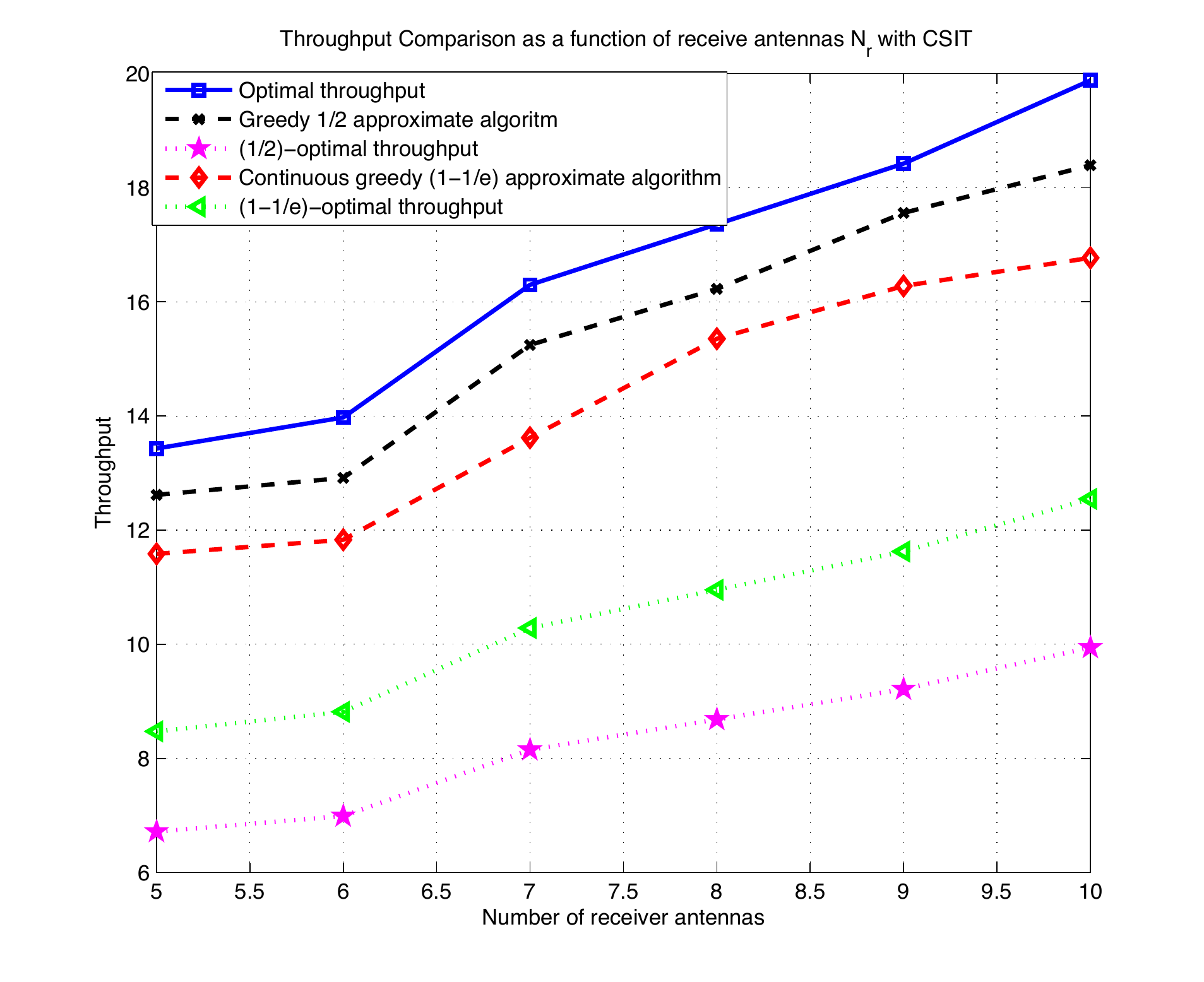}
\caption{Mutual information comparison with different algorithms for CSIT}
\label{fig:simCSIT}
\end{figure}

\section{Conclusion}
In this paper, we have considered a SWIPT system and found theoretical guarantees on a combinatorial problem of partitioning mobile antennas for information and power transfers to maximize the IT channel capacity under a circuit power constraint. We exploited the sub-modular property of the mutual information expression for both the CSIR and CSIT cases, as well as the matroidal property of the circuit power constraint. To the best of our knowledge, this is a novel approach in the area of SWIPT, and earlier approaches relied on relaxed problems, where each antenna is required to perform the dual role of information and power transfer mode, which is hard to realize in practice.
\bibliographystyle{IEEEtran}
\input{WPT_Ant_Assign.bbl}

%\bibliography{BibDesk_File}

\end{document}

%% file: header.tex
\newtheorem{theorem}{Theorem}

\newtheorem{definition}{Definition}

\newtheorem{lemma}{Lemma}
\newtheorem{corollary}{Corollary}

\def\E{\mathsf{E}}

\def\phi{\varphi}

\def\l{\left}
\def\r{\right}
\def\({\left(}
\def\){\right)}

\def\bydef{:=}

\def\bff{{\mathbf{f}}}
\def\bg{{\mathbf{g}}}

\def\bs{{\mathbf{s}}}

\def\bx{{\mathbf{x}}}

\def\b0{{\mathbf{0}}}

\def\bH{{\mathbf{H}}}
\def\bI{{\mathbf{I}}}

\def\bP{{\mathbf{P}}}
\def\bQ{{\mathbf{Q}}}
\def\bR{{\mathbf{R}}}
\def\bS{{\mathbf{S}}}

\def\cC{\mathcal{C}}

\def\cI{\mathcal{I}}

\def\cM{\mathcal{M}}

\def\cX{\mathcal{X}}

\newcommand{\tr}{\mathrm{tr}}

\newcommand{\nn}{\nonumber}